\documentclass[letterpaper, 10pt, conference]{IEEEtran}
\usepackage{float}
\usepackage{nicefrac}
\usepackage{slashbox}
\usepackage{verbatim}
\usepackage{tikz,pgfplots}
\usetikzlibrary{chains}
\pgfplotsset{compat=1.14} 
\usepackage{pgfplotstable}
\usepackage{pgf}
\usepackage{url}
\usepackage{bbm}
\usepackage{calc}
\usepackage{amsthm}
\usepackage{thmtools,thm-restate}
\usepackage{comment}
\usepackage{times}
\usepackage{amsthm}
\usepackage{amssymb}
\usepackage{color}
\usepackage{subcaption}
\usepackage{amsmath}
\usepackage{multirow}
\usepackage{soul}
\usepackage{algorithm}
\usepackage{algorithmic}
\usepackage{wrapfig}
\usepackage{epstopdf}
 \usepackage{booktabs}
\usepackage[font={small}]{caption}

\renewcommand{\baselinestretch}{0.95}
\newcommand{\specialcell}[2][c]{%
  \begin{tabular}[#1]{@{}c@{}}#2\end{tabular}}
  
\floatname{algorithm}{Algorithm}

\begin{document}
\bstctlcite{IEEEexample:BSTcontrol} % For limiting # of authors listed

%\makeatother
\title{Hyperprofile-based Computation Offloading for Mobile Edge Networks}

% \thanks{This material is based upon work supported by the National Science Foundation under Award Number: CNS-1659134. Any opinions, findings, and conclusions or recommendations expressed in this publication are those of the author(s) and do not necessarily reflect the views of the National Science Foundation.}

\author{\IEEEauthorblockN{Andrew Crutcher\IEEEauthorrefmark{1}, Caleb Koch\IEEEauthorrefmark{2}, Kyle Coleman\IEEEauthorrefmark{3}, Jon Patman\IEEEauthorrefmark{4}, Flavio Esposito\IEEEauthorrefmark{3}, Prasad Calyam\IEEEauthorrefmark{4}}
\IEEEauthorblockA{\IEEEauthorrefmark{1}Southeast Missouri State University, alcrutcher1s@semo.edu}
\IEEEauthorblockA{\IEEEauthorrefmark{2}Cornell University, cak247@cornell.edu}
\IEEEauthorblockA{\IEEEauthorrefmark{3}Saint Louis University, 
\{kylecoleman, espositof\}@slu.edu}
\IEEEauthorblockA{\IEEEauthorrefmark{4}University of Missouri, jpxrc@mail.missouri.edu; calyamp@missouri.edu}
}

\renewcommand\IEEEkeywordsname{Keywords}

\maketitle
\thispagestyle{empty}
\pagestyle{empty}

%%%%%%%%%%%%%%%%%%%%%%%%%%%%%%%%%%%%%%%%%%%%%%%%%%%%%%%%%%%%%%%%%%%%%%%%%%%%
\begin{abstract}
In recent studies, researchers have developed various computation offloading frameworks for bringing cloud services closer to the user via edge networks.  Specifically, an edge device needs to offload computationally intensive tasks because of energy and processing constraints. These constraints present the challenge of identifying which edge nodes should receive tasks to reduce overall resource consumption. We propose a unique solution to this problem which incorporates elements from Knowledge-Defined Networking (KDN) to make intelligent predictions about offloading costs based on historical data. Each server instance can be represented in a multidimensional feature space where each dimension corresponds to a predicted metric. We compute features for a ``hyperprofile'' and position nodes based on the predicted costs of offloading a particular task. We then perform a $k$-Nearest Neighbor ($k$NN) query within the hyperprofile to select nodes for offloading computation. This paper formalizes our hyperprofile-based solution and explores the viability of using machine learning (ML) techniques to predict metrics useful for computation offloading. We also investigate the effects of using different distance metrics for the queries. Our results show various network metrics can be modeled accurately with regression, and there are circumstances where $k$NN queries using Euclidean distance as opposed to rectilinear distance is more favorable. 

\begin{comment}
Given a set of network and device metrics regarding a task
A core research problem in edge computing is how and where to offload tasks. Often a set of metrics that represent 
This paper presents two advances on this problem by capturing 
Make sure we mention machine learning (ML) in the abstract.
The contribution of this research is twofold. First, 
\end{comment}

\end{abstract}


\theoremstyle{plain}
\newtheorem{theorem}{Theorem}
\newtheorem{proposition}{Proposition}
\newtheorem{lemma}{Lemma}
\newtheorem*{corollary}{Corollary}

\theoremstyle{definition}
\newtheorem{definition}{Definition}

\section{Introduction}
\label{sec:intro}

Internet of Things (IoT) technologies introduce the need for energy-aware and latency-sensitive management strategies to ensure reliable performance in resource constrained environments. One such situation is disaster incidents where first responders may be operating in areas with limited network and computing resources. Disaster response teams may also benefit from having sensor and visual data processed on site by utilizing computation offloading strategies to make optimal decisions based on energy or latency requirements of the user.

The computing environment of disaster response networks is similar to general edge computing networks in that we have pervasive computing infrastructure with multi-modal, multi-dimensional, and geospatially dispersed data sources that rely on a wide range of services (e.g. pedestrian tracking, facial recognition, location services)~\cite{vcc}. Typically, the main challenge of edge computing is concerned with how to execute these services on resource constrained devices such as mobile phones or other IoT devices. 

A popular and well-studied resolution to this challenge is computation offloading where resource intensive tasks are migrated to nearby cloud or edge servers with abundant computing resources. This is necessary because mobile devices are limited in terms of battery life, wireless communication, and computing capabilities~\cite{off}. Computation offloading is ideal because it typically results in lower energy consumption and processing time for the mobile user~\cite{iot}. Broadly speaking, computation offloading can offset the limitations of resource constrained mobile devices thereby offering a greater variety of services to the user~\cite{mal}. 

% Since SDN allows for centralized network analytics, it has become possible to apply machine learning (ML) techniques to automate network operations \cite{kdn}. 

The control mechanism for manging computation offloading has been a popular research topic and several offloading frameworks have been proposed~\cite{maui}~\cite{clonecloud}. Hence there is a desire to develop an intelligent, runtime offloading scheme~\cite{ada} to make decisions regarding when and how to offload. A new emerging paradigm known as Knowledge-Defined Networking (KDN) relies on Network Analytics (NA) and Software-Defined Networking (SDN) to efficiently learn stateful information about a network~\cite{kdn}. KDN makes use of NA to build a high-level model of the system known as the \textit{knowledge plane}~\cite{kp}.

Given the heterogeneous nature of edge networks, we can employ machine learning (ML) techniques to understand relationships between relevant variables that other analytical approaches may fail to capture. However, ML is only feasible if accurate training data is available. Traditionally, this is an issue as networks are inherently distributed systems and nodes have limited view and control of the network. However due to the development of SDN, the control and data planes can be decoupled which allows for a logically centralized control plane. SDN offers not only control of the network but also the ability to collect training data for the ML model.

In this paper, we aim to study the benefits of using KDN concepts to guide the design of an intelligent computation offloading framework. By intelligent we imply that our framework uses historical data to build a predictive model that can encode system and user dynamics that other deterministic heuristics may fail to capture. We design various network simulations in ns-3 in order to create a robust dataset that is used to train an ML model. We account for mobility by varying distances between the user and access points. The predictions from the model can then be used as input features in a multi-dimensional space we call the \textit{hyperprofile}. 

The hyperprofile consists of a set of nodes that correspond to physical machines which are positioned based on the predicted performance of that server for a given task. The user is represented in the feature space such that a query on the user's representation returns a set of nodes to which we offload the task. Our results indicate that the query method can play an integral role in scheduling tasks to offload to nodes. We provide a mathematical basis for why the Euclidean distance metric tends to favor nodes with a balanced trade-off between features.    
% We test our offloading scheme on an edge network designed using GENI server instances and local computing resources. We simulate offloading image data from the Google Glass \cite{glass} to a server instance over SFTP.  Each server will run two containers. One container is designed as an SFTP server while the other operates the facial recognition application. Both containers are linked to a shared volume on the host machine which allows for data transfer from that volume over SFTP and for our facial recognition  application to operate on images in that same volume. This setup allows our server instances to transfer data between each other to mitigate the constraint on user mobility.

The rest of the paper is organized as follows: Section \ref{sec:related} discusses previous work on computation offloading frameworks. In Section \ref{sec:prob} we formalize the offloading problem, Section \ref{sec:KDN} details our KDN-based model for selecting optimal nodes for offloading. Different query methods are discussed in Section \ref{sec:perfEval}, and Section \ref{sec:conclusion} concludes the paper.

\section{Related Work}
\label{sec:related}
\subsection{Computation Offloading in Mobile Edge Networks}
%Malmos
% mention use case of pictures with google glass and detecting objects
Mobile Edge Computing (MEC) can provide an energy-aware solution to computation offloading for IoT devices where energy conservation is more desirable than low-latency~\cite{enr}. The mobile users could take photographs of victim's faces and then perform facial recognition to identify the victim. The computation involved for facial recognition happens on either a core cloud or an edge computer near the user. However, many offloading frameworks only account for having one edge computer which limits the user in terms of mobility~\cite{enr}~\cite{maui}. They also rely on the user to decide what offloading strategy would be best for their situation. We posit that such decisions are most optimally handled by an ML model rather than someone unfamiliar with the structure and operation of networks (e.g. a first responder).

There are a number of existing offloading strategies. Yang \textit{et al}. studied the problem under multiple mobile device users sharing a common wireless connection to the cloud. Their solution uses a genetic heuristic algorithm and focuses on \textit{code partitioning} and deciding whether to offload each partition individually. Researchers in~\cite{iot} formulate the offloading problem as a multiple choice knapsack problem where one is trying to maximize bandwidth utilization subject to constraints such as battery life. An optimization approach is also taken in \cite{lag}, and their solution utilizes Lagrange multipliers.

It is difficult to adapt these solutions in real-time when some metrics may be unavailable or when the user needs change. A key feature of our solution is that it can be adapted based on the metrics available. This adaptability is achieved by effectively \textit{decoupling} the process of developing a network model (discussed in Section \ref{sec:KDN}) and incorporating the user into that model to suit his/her needs best (discussed in Section \ref{sec:perfEval}). 

\subsection{Knowledge--Defined Networking}

Mestres \textit{et al}. note in \cite{kdn} that ML models could work well with managing network behavior if the training data adequately represents the network itself. However, the authors remark it is unclear what constitutes representative training data in networking. This is left as an open research problem. On a fundamental level, our paper is motivated by the question of \textit{what characteristics of a network are relevant for developing ML models?} We believe our approach to this problem is unique because (1) we focus on the performance of \textit{feature selection} to develop a hyperprofile space and (2) then we apply $k$NN to find optimal destination nodes for offloading. We describe each formulation in depth in future sections (namely sections \ref{subsec:predictions} and \ref{subsec:kNN}).

\section{Problem Formulation}
% Also mention computation offloading and a formalization of this
\label{sec:prob}
Coordinating how computation is offloaded to edge servers can be seen as a job shop scheduling problem which is a popular problem in computer science literature \cite{jobScheduling2}. Job shop scheduling is concerned with optimally scheduling a set of jobs on machines with varying processing constraints. It minimizes some objective such as energy consumption or makespan (total time to process all jobs) \cite{jobScheduling1}. 

% Due to the heterogeneous processing capabilities of hardware in mobile edge networks and various operational constraints (e.g. energy consumption or latency), we can conceptualize the problem of offloading computation as an \textit{open shop} scheduling problem:
% $$
% Om \ | \ prmu \ | \ C_{max},  E_c
% $$
% where $Om$ defines the machine environment as an open shop where there are $m$ machines in parallel each with having different processing capabilities. The processing constraint $prmu$ refers to the requirement that the job queue for each machine follows the \textit{First In First Out} (FIFO) discipline. The third terms $C_{max}$ and $E_c$ are the objectives to be minimized and correspond to makespan (total processing time for all jobs) and energy consumption, respectively.

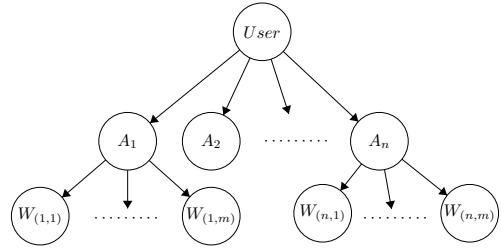
\begin{figure}[h]
\begin{center}
\resizebox{0.75\columnwidth}{!}{
\begin{tikzpicture}[scale=0.2]
\tikzstyle{every node}+=[inner sep=0pt]
\draw [black] (37.4,-15.6) circle (3);
\draw (37.4,-15.6) node {$User$};
\draw [black] (23.4,-27) circle (3);
\draw (23.4,-27) node {$A_1$};
\draw [black] (32.1,-27) circle (3);
\draw (32.1,-27) node {$A_2$};
\draw [black] (49.6,-27) circle (3);
\draw (49.6,-27) node {$A_n$};
\draw [black] (14.3,-34.8) circle (3);
\draw (14.3,-34.8) node {$W_{(1,1)}$};
\draw [black] (32.1,-34.8) circle (3);
\draw (32.1,-34.8) node {$W_{(1,m)}$};
\draw [black] (43.7,-34.5) circle (3);
\draw (43.7,-34.5) node {$W_{(n,1)}$};
\draw [black] (59,-34.5) circle (3);
\draw (59,-34.5) node {$W_{(n,m)}$};
\draw (41.1,-27) node {$\cdots\cdots\cdots$};
\draw (23.4,-35) node {$\cdots\cdots\cdots$};
\draw (51.4,-35) node {$\cdots\cdots\cdots$};
\draw [black] (21.12,-28.95) -- (16.58,-32.85);
\fill [black] (16.58,-32.85) -- (17.51,-32.71) -- (16.86,-31.95);
\draw [black] (25.63,-29) -- (29.87,-32.8);
\fill [black] (29.87,-32.8) -- (29.6,-31.89) -- (28.94,-32.64);
\draw [black] (35.07,-17.49) -- (25.73,-25.11);
\fill [black] (25.73,-25.11) -- (26.66,-24.99) -- (26.03,-24.21);
\draw [black] (36.14,-18.32) -- (33.36,-24.28);
\fill [black] (33.36,-24.28) -- (34.16,-23.76) -- (33.25,-23.34);
\draw [black] (39.59,-17.65) -- (47.41,-24.95);
\fill [black] (47.41,-24.95) -- (47.16,-24.04) -- (46.48,-24.77);
\draw [black] (47.75,-29.36) -- (45.55,-32.14);
\fill [black] (45.55,-32.14) -- (46.44,-31.82) -- (45.66,-31.2);
\draw [black] (51.95,-28.87) -- (56.65,-32.63);
\fill [black] (56.65,-32.63) -- (56.34,-31.74) -- (55.72,-32.52);
\draw [black] (38.33,-18.45) -- (40.17,-24.15);
\fill [black] (40.17,-24.15) -- (40.4,-23.23) -- (39.45,-23.54);
\draw [black] (23.4,-30) -- (23.4,-34);
\fill [black] (23.4,-34) -- (23.9,-33.2) -- (22.9,-33.2);
\draw [black] (50.13,-29.95) -- (50.87,-34.05);
\fill [black] (50.87,-34.05) -- (51.22,-33.17) -- (50.23,-33.35);
\end{tikzpicture}
}
\end{center}
\caption{\footnotesize{Job scheduling problem represented as a partition/aggregate application structure}}
\label{aggWork}
\end{figure}

We can further extend the concept of relating computation offloading to a job shop scheduling problem by viewing the task of offloading from the mobile user's perspective. This is achieved by representing the job scheduling problem as a directed, acyclic graph. Figure \ref{aggWork} shows one such depiction where the root node is the user who wishes to offload some computational tasks to available vertices where the edges between vertices can be weighted to represent the energy or latency cost associated with selecting that vertex. 

We adopt the popular partition/aggregate application structure which consists of a distributed set of aggregate and worker nodes~\cite{aggWorker}. Depending on the task partitioning, a user can forward data to the nearest aggregator that then schedules which worker nodes receive which jobs. Each aggregate node $A_i$ can be seen as an independent job shop that receives a set of jobs $J$ and in turn schedules them to be processed on a set of available worker machines $W$. In this paper, we focus on the first level of offloading from the user to the aggregator. 

% alpha - machine environment
% beta - processing characteristics or constraints
% gamma - objective function

\section{KDN Based Offloading Framework}
\label{sec:KDN}
An overview of our solution framework is illustrated in Figure \ref{fig:KDN_CO_Model}. The first step to our solution involves collecting network features for training an ML model. This part is important for developing the hyperprofile. 
\begin{figure}[H]
\vspace{-3mm}
\centering
\resizebox{\columnwidth}{!}{
\includegraphics[width=0.9\textwidth]{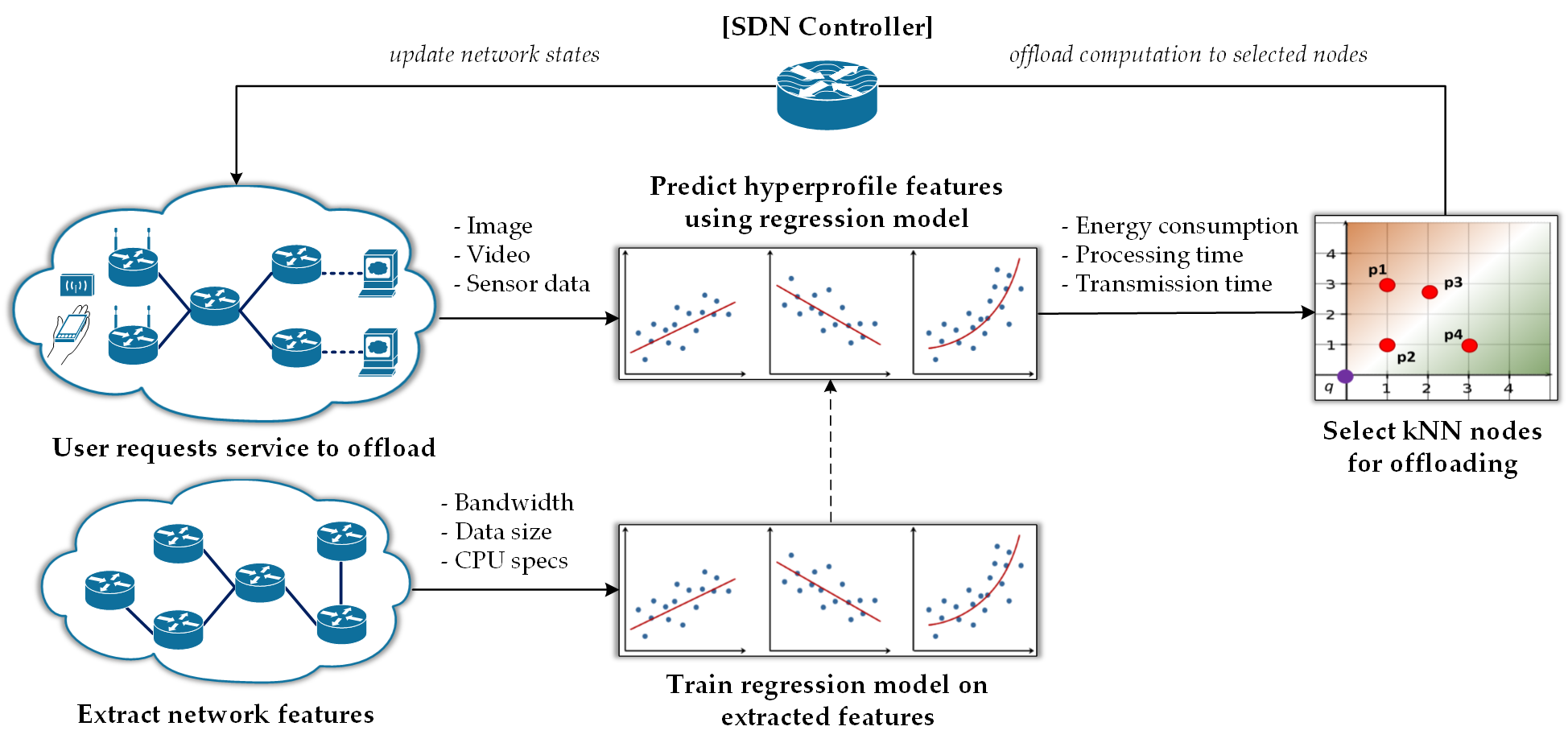}
}
%\vspace{-4mm}
\caption{\footnotesize{
Overview of our proposed hyperprofile-based computation offloading framework.  Temporally, the process begins with training a model. Then, as the user makes requests for offloading the model is applied to develop a hyperprofile within which $k$NN queries can be made to select nodes. The SDN controller servers to initiate the offloading process
}}
\label{fig:KDN_CO_Model}
\vspace{-3mm}
\end{figure}

\subsection{Data Collection}
We ran multiple simulations using ns-3 \cite{ns3}, a discrete-event network simulator that is available for research. Specifically, we ran our simulations using version 3.26 of ns-3 on Ubuntu 16.04 in parallel using GNU parallel \cite{Tange2011a}. We chose ns-3 as our network simulator for its tracing subsystem, energy framework for Wi-Fi devices, and ease of running the same simulation with modified program parameters. These capabilities allow us to generate training data. We simulated sending data between a wireless device and access point while varying bandwidth and total data sent. We measured energy consumed by the wireless device and the time that passed from the first packet being sent by the source to the last packet being received.

The simulations provide a basis for developing the hyperprofiles for the servers. Our goal is to predict energy consumption and transmission time from bandwidth and data size. Figure \ref{fig:3d_2} depicts the relationship between energy consumption, transmission time, bandwidth, and data size as a scatter plot of the raw data.

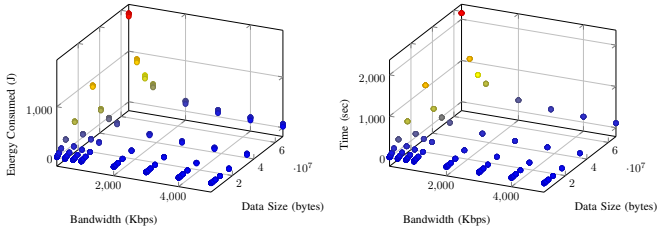
\begin{figure}
    \centering
    \begin{minipage}{0.5\columnwidth}
        \centering
        \resizebox{\columnwidth}{!}{
            \begin{tikzpicture}
            \begin{axis} [
              grid = major,
              xlabel={Bandwidth (Kbps)},
              ylabel = {Data Size (bytes)},
              zlabel = {Energy Consumed (J)},
            ]
                \addplot3+[only marks,scatter] table[y={DataSize}, x={Bandwidth} , z={EnergyConsumed}, col sep=comma] {LatexData3d.txt};
                %\addplot3[color=black, samples=30, range=655536:67108864, thick, smooth, x=250]
                %  (250,{y},{2.6*y*pow(10.0,-6.0)});
            \end{axis}
            \end{tikzpicture}
        }
        %\caption{first figure}
    \end{minipage}\hfill
    \begin{minipage}{0.5\columnwidth}
        \centering
        \resizebox{\columnwidth}{!}{
            \begin{tikzpicture}
            \begin{axis} [
              grid = major,
              xlabel={Bandwidth (Kbps)},
              ylabel = {Data Size (bytes)},
              zlabel = {Time (sec)},
            ]
                \addplot3+[only marks,scatter] table[y={DataSize}, x={Bandwidth} , z={TimeSending}, col sep=comma] {Lat3dDataTimeSending.csv};
                %\addplot3[color=black, samples=30, range=655536:67108864, thick, smooth, x=250]
                %  (250,{y},{2.6*y*pow(10.0,-6.0)});
            \end{axis}
            \end{tikzpicture}
           }
        %\caption{second figure}
    \end{minipage}
   \caption{\footnotesize{Three-dimensional plot of the energy consumption and time sending data gathered from the simulations using ns-3 where the amount of data to send and bandwidth of the connection was varied across simulations. The key point is that there is an exponential relationship between energy/time and bandwidth for a fixed data size while there is a linear relationship between energy/time and data size for a fixed bandwidth.}}
   \label{fig:3d_2}
\end{figure}

\subsection{Predictive Analytics}
\label{subsec:predictions}
We found a multistep regression was the most appropriate approach to achieving our desired prediction. Specifically, we use bandwidth to predict the linear regression line that models energy consumption and data size. Our results show that the slope of such lines are exponential with respect to bandwidth. The same relationships apply to predicting time, as depicted in Figure \ref{fig:3d_2}.

Formally, we can represent the predicted variable as a linear function $$f_b(d_s) = m(b) d_s + c(b)$$ where the slope $m(b)$ and the y-intercept $c(b)$ are functions of some bandwidth $b$. Our results show that $m$ is exponential and can be defined accurately from historical data. Note that for energy consumption $c(b)=0$ because the amount of energy consumed when sending 0 bytes of data will always be 0, regardless of bandwidth.

%Figure \ref{fig:slope} plots $m(b)$ when $f$ is energy consumption. The plot for when $f$ is transmission time is also exponential and appears the same visually.

We perform linear regression for each fixed value of bandwidth between energy consumption and data size. In Figure \ref{fig:3d_2} these lines can be seen by connecting scatter points along a particular bandwidth value. We predict the slope of these lines using bandwidth values. The results of these predictions are shown in Table \ref{tab:accur}. The regression resulted in various curves of best fits which are given explicitly in the table.

%The results of these predictions are shown in Figure \ref{fig:slope}. Notice that the slope decreases exponentially as a function of bandwidth. We perform regression to obtain the curve of best fit which is given explicitly in Table \ref{tab:accur} and is shown graphically in Figure \ref{fig:slope}. We repeat this process for time data.

Table \ref{tab:accur} also reports the $R^2$ value of the models used to compute $m$ and $c$. It also depicts the $k$-fold cross-validation score for $k=10$. This score validates that for a fixed bandwidth the relationship between the predicted value and the data size can be modeled with a linear function.

\begin{table}[H]
\centering
\bgroup
\def\arraystretch{1.5} %  1 is the default, adds vertical padding
\resizebox{0.75\columnwidth}{!}{
\begin{tabular}{|l||c|c|}\hline
%\backslashbox{Input}{Output}
 &\makebox{Energy Consumption ($e_c$)}&\makebox{Time ($t$)}\\\hline\hline
Bandwidth ($b$)   & \specialcell{$m_1=0.015b^{-1.13}$\\ $R^2=0.997$} & \specialcell{$m_2=\nicefrac{8.04\cdot 10^6}{b}$\\ $R^2=1$\\$c=222873e^{0.0004b}$\\$R^2=0.918$}\\\hline
Data Size ($d_s$) & \specialcell{$e_c = m_1d_s$\\Cross-validation: 0.99}      & \specialcell{$t = m_2d_s+c$\\Cross-validation: 0.99}\\\hline
\end{tabular}
}
\egroup
\caption{\footnotesize{Accuracy of predictions. In all cases, the $R^2$ is greater than $0.9$. The lowest score come from predicting the y-intercept of transmission time; to obtain a higher accuracy we would need to collect more data}}
\label{tab:accur}
\end{table}

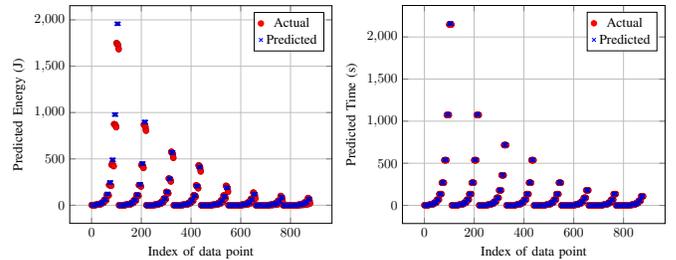
\begin{figure}[h]
\centering
    \begin{minipage}{0.5\columnwidth}
        \centering
        \resizebox{\columnwidth}{!}{
        \begin{tikzpicture}
        %\selectcolormodel{gray}
        \begin{axis}[
                xlabel=Index of data point, % label x axis
                ylabel=Predicted Energy (J), % label y axis
                grid=major,
                legend pos = north east,
            ]
            \addplot[
            scatter=true,
            only marks,
            point meta = 2,
            color=red,
            ] table[x=Index, y=EnergyActual, col sep=comma] {predictionResultsEnergy.csv};
            \addlegendentry{Actual}
            \addplot[
              scatter=true,
              only marks,
              point meta = 1,
              color=blue,
              mark=x,
            ] table[x={Index}, y=EnergyPredicted, col sep=comma] {predictionResultsEnergy.csv};
            \addlegendentry{Predicted}
            %\legend{Predicted, Actual}
            \end{axis}
        \end{tikzpicture}
        }
    \end{minipage}\hfill
    \begin{minipage}{0.5\columnwidth}
        \centering
        \resizebox{\columnwidth}{!}{
    \begin{tikzpicture}
    \begin{axis}[
            xlabel=Index of data point, % label x axis
            ylabel=Predicted Time (s), % label y axis
            grid=major,
            legend pos = north east,
        ]
        \addplot[
        scatter=true,
        only marks,
        point meta = 2,
        color=red,
        ] table[x=Index, y=TimeActual, col sep=comma] {predictionResultsTimeSec.csv};
        \addlegendentry{Actual}
        \addplot[
          scatter=true,
          only marks,
          point meta = 1,
          color=blue,
          mark=x,
        ] table[x={Index}, y=TimePrediction, col sep=comma] {predictionResultsTimeSec.csv};
        \addlegendentry{Predicted}
        %\legend{Predicted, Actual}
        \end{axis}
\end{tikzpicture}
}
\end{minipage}
\caption{\footnotesize{Analysis of our prediction models. Each point on the $x$-axis represents a trial of a specific bandwidth and data size. The trials also varied physical distance, and even though our model did not account for distance, it still performed well -- as shown by the overlap between the actual points and the predicted points.}}
\label{fig:predModel}
\end{figure}

In another set of simulations, we varied the physical distance between nodes from 10 m to 100 m. We avoided larger distances because we wanted to focus on scenarios where packet loss is not part of the issue of offloading. Testing our prediction model on this data shows that the variation of distances contributes an insignificant amount of error. The results of our predictions are given in Figure \ref{fig:predModel} where the x-axis plots the index of the data point and the y-axis represents dependent variable (both predicted and actual). The data clusters naturally into groups based on bandwidth, the first group being from index 0 to $\sim 100$. The energy consumption/transmission time increase within the groups because data size is increasing. Notice that the error is worst for smaller indices. This shows that the error is worst for small bandwidths and large data sizes. This observation could be attributed to the way error compounds with larger transmission times (since the largest transmission time occurs with a small bandwidth and a large data size). The main point is that in the vast majority of cases, we can almost exactly predict energy consumption and transmission time.

\subsection{Hyperprofiles for Efficient Offloading}
Our regression analysis shows that from historical data, we can develop accurate models of network features. A natural question is how we can leverage such models to make intelligent offloading decisions. Given that we are using ML to predict network metrics specific to a server, an intuitive answer to this question involves representing the available servers in a ``feature space'' where each dimension of the space is a modeled metric (e.g. energy consumption). The user device (or aggregator) can then be intelligently places in the feature space such that its position relative to the servers' positions conveys meaningful information. For our metrics, we want to minimize energy consumption and transmission time, so the user will always be represented by the origin. That way, distance from the origin conveys a level of desirability (i.e. the farther a node is from the origin the less desirable it is). This representation is particularly useful because if the device application needs to partition a task into $k$ parts for computation offloading, it can perform a $k$NN query on the origin to get a set of server points in the feature space with the ``optimal'' resources for processing the task.

The metrics or ``profiles'' that represent available servers in the feature space do not necessarily have to be predicted values. They could be specifications of the servers themselves such as processor clock speed. To help distinguish between the various profiles, we introduce the following terminology.

\begin{definition}[Base Profile]
A base profile for an edge network consists of points in a feature space where each point represents a unique server instance, and each dimension of the point represents a deterministic metric. Such metrics may include internal instance specifications (e.g. internal memory), characteristics of the network (e.g. bandwidth), or real-time metrics (e.g. CPU load). 
\end{definition}

\begin{definition}[Hyperprofile]
A hyperprofile is similar to the base profile except each dimension represents a predicted metric. One example of such a metric is the estimated time to receive and transmit a data packet from an external host.
\end{definition}

The main idea behind the different metrics is that they indicate a level of ``fitness'' of each server which can be quantitatively compared with the user device needs and specifications. Different profiles may be appropriate for different tasks and different scenarios. For example, in cases where no historical data is available one may opt to use the base profiles. An interesting direction for future research is to investigate the trade-off between the various profiles and whether one is significantly more useful than the others. It may even be helpful to combine the profiles into a \textit{hybridprofile}. Regardless, the use of profiles is advantageous because it reduces the computation offloading problem to a $k$NN query.

\subsection{Queries on hyperprofile features}
\label{subsec:kNN}

We developed the idea of the hyperprofiles with the intention of performing $k$NN queries on the user device to obtain a set of server instances to which we offload. Nonetheless, it should be noted that different types of queries may employ different search algorithms. Moreover, it need not be the case that every dimension is a metric that we want to minimize. And hence, it may not be the case that the user is always represented by the origin.

Regardless, for our model, $k$NN was the most appropriate algorithm because it returns the points most ``optimal'' relative to the user. Formally, $k$NN returns a set of points for a query point $q$ such that $p\in k\text{NN}(q) \text{ iff } |\{j\in P: d(j, q) < d(p, q)\}| <k$ where $d$ is a predefined distance metric. Often the distance metric is Euclidean which means, in our case, $k$NN($\vec{0}$) returns the points $p_i=(x_i,y_i)$ such that the values of $x_i^2+y_i^2$ are minimal. When $x$ is energy consumption and $y$ is transmission time, offloading to the servers represented by points in $k$NN($\vec{0}$) minimizes energy loss and latency. Other approaches (e.g. Chen \cite{game}) minimize $x+y$ which is effectively the same as performing $k$NN where $d$ is rectilinear distance. It may seem that minimizing one may be the same as minimizing the other, but as Table \ref{tab:distance} depicts, this is not the case. The rest of this section is dedicated to discussing the difference between these two distance metrics.

\begin{table}[H]
\centering
\caption{\footnotesize{Example scenario where a query of one point returns different sets. If using Euclidean distance $k$NN($\vec{0}$) = $\{p_2\}$, whereas if using rectilinear distance, $k$NN($\vec{0}$) = $\{p_1, p_4\}$. Notice that $|x-y|$ is smaller for $p_2$ than for $p_1$. This relationship is explored further by Proposition \ref{prop:diff} (see Appendix)}}
\label{tab:distance}
\bgroup
\def\arraystretch{1.5} %  1 is the default, adds vertical padding
\begin{tabular}{|l||c|c|}\hline
\backslashbox{Point}{Distance metric}
 &\makebox{Euclidean}&\makebox{Rectilinear}\\\hline\hline
$p_1=(0.219, 0.371)$   & $0.431$ & $0.59$ \\\hline
$p_2 = (0.233, 0.361)$ &  0.429  & 0.594\\\hline
\end{tabular}
\egroup
\end{table}

\section{Different distance metrics for $k$NN queries}
\label{sec:perfEval}

\usetikzlibrary{patterns}
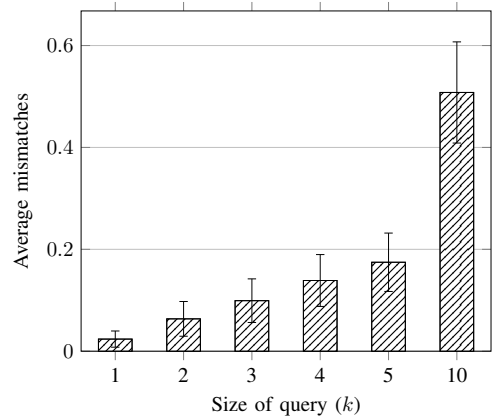
\begin{figure}[h]
\centering
\resizebox{0.75\columnwidth}{!}{
\begin{tikzpicture}
    \begin{axis}[
        %width  = 0.85*\textwidth,
        %height = 8cm,
        %major x tick style = transparent,
        ybar=\pgflinewidth,
        bar width=16pt,
        ymajorgrids = true,
        ylabel = {Average mismatches},
        xlabel = {Size of query ($k$)},
        symbolic x coords={1, 2, 3, 4, 5, 10},
        xtick = data,
        ymin=0,
    ]
        \addplot[style={black,fill=white,mark=none},postaction={pattern=north east lines}, error bars/.cd, y dir=both,y explicit]
            coordinates {
            (1, 0.0238) +- (0.0, 0.0159)
            % +=(0,4) -= (0,2)
            %+- (0.033)
            (2,0.0635)  +- (0.0,0.0341)
            (3,0.0992)  +- (0.0, 0.0428)
            (4, 0.1388) +- (0.0, 0.0509)
            (5, 0.1746)  +- (0.0,0.0573)
            (10, 0.5079)  +- (0.0,0.0993)
            };
        %\legend{}
    \end{axis}
\end{tikzpicture}
}
\caption{\footnotesize{Average mismatches between queries with 95\% confidence intervals. These results come from randomly generated bandwidth and data size values where we varied the size of the datasets and the size of the query (i.e. the value of $k$). We found that larger queries have a higher probability of mismatch}}
\label{fig:mismatches}
\end{figure}

We performed a set of simulations on randomly generated hyperprofiles to evaluate how often a $k$NN query would give different results for different distance metrics. Our simulations consisted of 2 dimensional hyperprofiles ranging in size from 250 points to 5000 points. Each point was computed from randomly generated bandwidth and data size values. Bandwidth ranged from 250 Kbps to 15 Mbps while data size ranged from 60 kB to 250 MB. For each hyperprofile we performed $k$NN queries where $k\in \{1, 2, 3, 4, 5, 10\}$. The queries were performed with both Euclidean and rectilinear metrics and the mismatches between the two methods is shown in Figure \ref{fig:mismatches}. 

Throughout the simulations, we kept track of the points that were mismatched, and in all the cases we noticed that the Euclidean distance queries returned points with minimal differences between $x$ and $y$. In other words, it returned points closer to the line $y=x$. After formalizing our observation, we found this was an inherent property of the mismatched points. Thus, as Proposition \ref{prop:diff} (See Appendix) states, if there is a mismatch in the minimal Euclidean distance and the minimal rectilinear distance between points then the distance between the coordinates of the minimal point based on Euclidean distance is less than the distance between the coordinates of the minimal point based on rectilinear difference.

% \begin{restatable}{proposition}{mainProp}
% \label{prop:diff}
% If \emph{(1)} $x_1,y_1, x_2, y_2\geq 0$, \emph{(2)} $x_1^2+y_1^2<x_2^2+y_2^2$, and \emph{(3)} $x_2+y_2 < x_1 + y_1$ then $|x_1-y_1| < |x_2-y_2|$.
% \end{restatable}

% \begin{proposition}
% \label{prop:diff}
% If $x_1,y_1, x_2, y_2\geq 0$, $x_1^2+y_1^2<x_2^2+y_2^2$, and $x_2+y_2 < x_1 + y_1$ then $|x_1-y_1| < |x_2-y_2|$.
% \end{proposition}

Notice that the first condition is that the coordinates are nonnegative. Without this condition, a simple counterexample such as $(x_1,y_1)=(3, 0)$ and $(x_2, y_2)=(-3, 1)$ would falsify the proposition. The condition is reasonable since often the hyperprofile dimensions represent nonnegative features of the network or user device such as energy consumption or latency.

Ultimately, the key point of Proposition \ref{prop:diff} is that a Euclidean distance metric will favor points with a more balanced tradeoff between the network features represented by the coordinates. Moreover, the difference becomes more pronounced as $k$ becomes larger. Hence, whether this tradeoff is favorable depends on both the size of the edge network and the types of features with which the user is concerned.

\section{Conclusion}
\label{sec:conclusion}
%In this paper, we outlined a framework for solving computation offloading problems by creating a hyperprofile based on different metrics. A $k$NN query on a point that represents an optimal solution in this hyperprofile gives the desired nodes. Our solution is unique because it can be easily adapted to different network metrics based on user needs. Our approach leverages the benefits of the knowledge plane of KDN to make meaningful inferences from raw data. We showed these inferences can come from an ML model trained on simulation data. Finally, our solution uses $k$NN to find the nodes to which the task is offloaded. We investigated the effects of using different distance metrics for $k$NN, formalized when there is a mismatch, and interpreted what this meant in terms of using the hyperprofile to offloading tasks.

% Remove some of the future work, add some to conclusion; emphasize KDN

In this paper, we outlined a framework for computation offloading in disaster network applications by creating a hyperprofile of network edge nodes using metrics such as task processing time. Our solution is unique because it can be easily adapted to different network metrics based on user needs. It decouples the problem of modeling the available network metrics from the problem of identifying user needs while making use of ML models. Effectively, we have described a way of \textit{encoding} available server instances in a multidimensional space to facilitate effective queries. For future work, we first plan to implement a testbed to compare a hyperprofile-based offloading scheme to standard offloading schemes. Another area of interest is expanding this abstract concept of a hyperprofile to other areas such as routing or trust management among others.

\appendix
%\mainProp*
\begin{proposition}
\label{prop:diff}
If \emph{(1)} $x_1,y_1, x_2, y_2\geq 0$, \emph{(2)} $x_1^2+y_1^2<x_2^2+y_2^2$, and \emph{(3)} $x_2+y_2 < x_1 + y_1$ then $|x_1-y_1| < |x_2-y_2|$.
\end{proposition}
\begin{proof}
First note that $x_2\not=x_1$ since if they were equal then we would have $y_1^2<y_2^2$ and $y_2<y_1$, a contradiction. The same reasoning implies $y_2\not=y_1$. Now assume without loss of generality that $x_1\geq y_1$. We can write (3) as
$$
(x_2-x_1) + (y_2 - y_1) < 0.
$$
We deal with the problem in cases based on whether $(x_2-x_1)$ or $(y_2 - y_1)$ is negative. For the first case, assume $x_2-x_1<0$. Now write (2) as
$$
(y_1 - y_2)(y_1+y_2) < (x_2 - x_1)(x_2+x_1).
$$
Since $x_2-x_1<0$ and the sums are positive, we must have $y_1 - y_2 <0$. Hence, $y_2-y_1 >0$. Now write (2) as
\begin{equation}
(x_1-x_2)(x_1+x_2) < (y_2-y_1)(y_2+y_1).
\end{equation}
and write (3) as $y_2 - y_1 < x_1 - x_2$. Combining these we get
$$
(x_1-x_2)(x_1+x_2) < (x_1 - x_2)(y_2+y_1).
$$
Since $x_1 - x_2$ is positive we can divide it out to get $x_1 + x_2 < y_2 + y_1$ which we rearranging to get,
$$
|x_1 - y_1| = x_1 - y_1 < y_2 - x_2 \leq |y_2 - x_2|
$$
by our original assumption that $x_1 \geq y_1$. Now for case (2) assume that $(y_2 - y_1) < 0$. From Equation (1), this means $x_1 - x_2 <0$. Now combining these with our assumption that $x_1 \geq y_1$ we have $x_2 > x_1 \geq y_1 > y_2>0.$ From this, We can write
$$
|x_2-y_2|\geq x_2 - y_2 > x_1 - y_1 = |x_1 - y_1|.
$$

% Now we can subtract $y_2$ to get
% $$
% x_2-y_2 > x_1 - y_2 \geq y_1 - y_2.
% $$
%---------------------------------------------
% Write (2) as
% $$
% (x_1-x_2)(x_1+x_2) < (y_2-y_1)(y_2+y_1).
% $$
% From (3) we have that $(x_1-x_2)>(y_2-y_1)$, so we can write
% $$
% (x_1-x_2)(x_1+x_2) < (x_1-x_2)(y_2+y_1).
% $$
% Given this expression and that $x_1+x_2$ and $y_2 + y_1$ are nonnegative and that $(x_1-x_2) + (y_1-y_2)<0$ by (3), it must be the case that $(x_1-x_2)<0$. Thus we divide out $(x_1-x_2)$ and switch signs to get
% $$
% (x_1+x_2) > (y_2+y_1).
% $$
% This implies
% $$
% |x_2-y_2|\geq x_2 - y_2 > y_1 - x_1 = |y_1-x_1|.
% $$
% The last equality holds because our assumption that $y_1\geq x_1$.
\end{proof}

\section*{Acknowledgements}
This material is based upon work supported by the National Science Foundation under Award Number: CNS-1659134 [REU Site], CNS-1647182, and CNS-1647084. Any opinions, findings, and conclusions or recommendations expressed in this publication are those of the author(s) and do not necessarily reflect the views of the National Science Foundation.

\bibliographystyle{IEEEtran}
\bibliography{bibitex}
\end{document}